\documentclass[12pt]{article}
\usepackage[]{times}
\usepackage[]{amsmath}
\usepackage[]{amsthm}
\usepackage[colorlinks]{hyperref}

\newtheorem{proposition}{Proposition}

\newtheorem*{theorem*}{Theorem}
\newtheorem*{lemma*}{Lemma}
\newtheorem*{proposition*}{Proposition}

\title{Gosper's algorithm and Bell numbers}
\author{Robert Dougherty-Bliss\thanks{Rutgers University, Department of Mathematics}}
\date{\today}

\begin{document}

\maketitle

\begin{abstract}
    \noindent Computers are good at evaluating finite sums in closed form, but
    there are finite sums which do not have closed forms. Summands which do not
    produce a closed form can often be ``fixed'' by multiplying them by a
    suitable polynomial. We provide an explicit description of a class of such
    polynomials for simple hypergeometric summands in terms of the Bell
    numbers.
\end{abstract}

\noindent Evaluating the partial sums of sequences which are products and
quotients of polynomials, binomial coefficients, factorials, and so on is a
major theme in combinatorics, discrete probability, and computer science. The
main tool in this area is Bill Gosper's marvelous hypergeometric summation
algorithm \cite{gosper}.

A sequence $f(k)$ is \emph{hypergeometric} (or a \emph{hypergeometric term})
provided that the consecutive quotient $f(k + 1) / f(k)$ is a rational function
in $k$. Gosper's algorithm completely solves the problem of hypergeometric
summation in one variable. It constructively determines when $\sum_k f(k)$
itself is hypergeometric term. We call such hypergeometric terms ``Gosper
summable.''

Unfortunately, many hypergeometric terms are \emph{not} Gosper summable. For
example, we cannot fill in the following blank with a hypergeometric term:
\begin{equation*}
    \sum_{k = 0}^n \frac{1}{k!} = \underline{\hspace{2in}}.
\end{equation*}
However, an upshot of Gosper's algorithm is that we can often \emph{tweak}
non-Gosper summable terms to make them summable. For example, while $1 / k!$ is
not Gosper summable, the term $(k - 1) / k!$ \emph{is}. In fact, lots of
multiples of $1 / k!$ are Gosper summable:
\begin{align*}
    \sum_{k = 0}^n \frac{k - 1}{k!} &= -\frac{1}{n!} \\
    \sum_{k = 0}^n \frac{k^2 - 2}{k!} &= -\frac{n + 2}{n!} \\
    \sum_{k = 0}^n \frac{k^3 - 5}{k!} &= -\frac{n^2 + 3n + 5}{n!} \\
    \sum_{k = 0}^n \frac{k^4 - 15}{k!} &= -\frac{n^3 + 4n^2 + 9n + 15}{n!}.
\end{align*}
This list suggests that there exists a sequence of integers $b(d)$---beginning
$1, 2, 5, 15$---such that $(k^d - b(d)) / k!$ is Gosper summable. This turns
out to be true. Even better, $b(d)$ turns out to be the $d$th Bell number, the
number of partitions of $d$ elements into any number of nonempty subsets.

Indeed, there is a similar statement for hypergeometric terms of the form $z^k
a^{\overline{k}}$ and $z^k / a^{\overline{k}}$ for constant $z$, where
$a^{\overline{k}} = a (a + 1) \cdots (a + k - 1)$ is the \emph{rising
factorial}. Specifically, there are explicit exponential generating functions
$g_{a, z}(x)$ and $f_{a, z}(x)$ such that $(k^d - c(d)) z^k a^{\overline{k}}$
is Gosper summable iff $c(d)$ is the coefficient on $x^d / d!$ in $g_{a,
z}(x)$, and the analogous statement for $z^k / a^{\overline{k}}$ and $f_{a,
z}(x)$. These generating functions happen to be related to the famous
exponential generating function for the Bell numbers $B(x) = e^{e^x - 1}$. Our
goal is to explain and prove these facts.

The remainder of this paper is organized as follows. Section~\ref{sec:gosper}
gives a quick overview of Gosper's algorithm. Section~\ref{sec:summability}
establishes the summability results and gives the explicit generating
functions. Section~\ref{sec:explicit} shows how to explicitly evaluate a
special case of these sums in terms of well-known integer sequences.
Section~\ref{sec:bell} explains how these generating functions are related to
the Bell numbers.

\section{Gosper's algorithm}%
\label{sec:gosper}

This section provides a brief overview of Gosper's algorithm. For more details,
see \cite{gosper} or \cite{ab}.

A sequence $f(k)$ is \emph{hypergeometric}, or a \emph{hypergeometric term},
provided that $f(k + 1) / f(k)$ is a rational function in $k$.

Every rational function $R(k)$ can be decomposed as
\begin{equation*}
    R(k) = \frac{a(k)}{b(k)} \frac{c(k + 1)}{c(k)},
\end{equation*}
where $a$, $b$, and $c$ are polynomials in $k$ which satisfy $\gcd(a(k), b(k +
i)) = 1$ for all nonnegative integers $i$. This is called the polynomial normal
form of $R(k)$. If $f(k)$ is hypergeometric and $f(k + 1) / f(k)$ has
polynomial normal form
\begin{equation*}
    \frac{f(k + 1)}{f(k)} = \frac{a(k)}{b(k)} \frac{c(k + 1)}{c(k)},
\end{equation*}
then we call $a / b$ the \emph{kernel} of $f$, and $c$ the \emph{shell} of $f$.
Note that
\begin{equation*}
    f(k) = z c(k) \prod_{j = 0}^{k - 1} (a(j) / b(j))
\end{equation*}
for some constant $z$. For this reason, we sometimes call $c(k)$ the
``polynomial part'' of $f(k)$ and the remaining product the ``purely
hypergeometric part.''

Gosper's algorithm amounts to the following theorem.

\begin{theorem*}
    The hypergeometric term $f(k)$ with polynomial normal form $(a, b, c)$ is
    Gosper summable if and only if there is a polynomial solution $x(k)$ to
    \begin{equation}
        \label{gosperEqn}
        x(k + 1) a(k) - x(k) b(k - 1) = c(k).
    \end{equation}
    In that case,
    \begin{equation*}
        \sum_k f(k) = \left( \frac{x(k) b(k - 1)}{c(k)} \right) f(k).
    \end{equation*}
\end{theorem*}

For example, the term ratio of $f(k) = 1 / k!$ has polynomial normal form
\begin{equation*}
    \frac{f(k + 1)}{f(k)} = \frac{1}{k + 1}
\end{equation*}
with $(a, b, c) = (1, k + 1, 1)$. Therefore $f(k)$ is summable if and only if
$x(k + 1) - x(k) k = 1$ has a polynomial solution $x(k)$, which it does not. On
the other hand, the term ratio of $g(k) = (k - 1) / k!$ has polynomial normal
form
\begin{equation*}
    \frac{g(k + 1)}{g(k)} = \frac{1}{k + 1} \frac{k}{k - 1}
\end{equation*}
with $(a, b, c) = (1, k + 1, k - 1)$. Therefore $g(k)$ is summable if and only
if $x(k + 1) - x(k) k = k - 1$ has a polynomial solution $x(k)$, which it does,
namely $x(k) = -1$. In addition,
\begin{equation*}
    \sum_k \frac{k - 1}{k!} = -\frac{k}{k!} = -\frac{1}{(k - 1)!}.
\end{equation*}

\section{Summability}%
\label{sec:summability}

For pedagogical purposes, let us first prove the following proposition.

\begin{proposition}
    The term
    \begin{equation*}
        \frac{k^d - b(d)}{k!}
    \end{equation*}
    is Gosper summable if and only if $b(d)$ is the $d$th Bell number.
\end{proposition}

\begin{proof}
    The consecutive term ratios of $1 / k!$ are $1 / (k + 1)$, which has
    polynomial normal form $(1, k + 1, 1)$. Setting $x_d(k) = -k^d$ in
    \eqref{gosperEqn} shows that
    \begin{equation*}
        p_d(k) = k^{d + 1} - (k + 1)^d
    \end{equation*}
    is a sequence of polynomials such that $p_d(k) / k!$ is Gosper summable.
    Since the degree of $p_d(k)$ is exactly $d + 1$, these polynomials are
    linearly independent, and therefore form a basis for the set of \emph{all}
    polynomials $p(k)$ such that $p(k) / k!$ is Gosper summable. Our
    proposition amounts to the claim that $k^d + b(d)$ is a different basis for
    this space.

    The degrees of the $p_d(k)$ start at $1$ and increase by $1$ every step, so
    by subtracting appropriate multiples of previous terms, we can cancel every
    power of $k$ in $p_d(k)$ except the leading term and the constant. That is,
    there \emph{is} a basis of the form $k + c(1), k^d + c(2), \dots$, obtained
    by a linear operation on the $p_d(k)$. In particular, since
    \begin{equation*}
        p_d(k) = k^{d + 1} - \sum_j {d \choose j} k^j,
    \end{equation*}
    the correct multiples to subtract are as follows:
    \begin{equation}
        \label{basis1}
        k^{d + 1} + c(d + 1) = p_d(k) + \sum_{j > 0} {d \choose j} (k^j + c(j)).
    \end{equation}
    If we set $c(0) = -1$ and look at the constant term of both sides, we obtain
    \begin{equation*}
        c(d + 1) = \sum_j {d \choose j} c(j).
    \end{equation*}
    This implies $c(d) = -b(d)$, where $b(d)$ is the $d$th Bell number, since
    the Bell numbers satisfy the same recurrence and begin with $1$ rather than
    $-1$.
\end{proof}

The above outline carries over nearly verbatim to other simple hypergeometric
terms. A slight difference is that, most of the time, the sequence $c(d)$ is
not well-known, and we have to settle for an explicit exponential generating
function. The following propositions neatly summarize the results.

\begin{proposition}
    If $z \neq 0$ and $a$ is not a nonpositive integer, then the hypergeometric
    term $(k^d - c(d)) z^k / a^{\overline{k}}$ is Gosper summable if and only
    if
    \begin{equation*}
        c(d) = [x^d / d!] \exp(-z - (a - 1) x + z e^x) = [x^d / d!] f_{a, z}(x).
    \end{equation*}
\end{proposition}

\begin{proof}
    The consecutive term ratios of $z^k / a^{\overline{k}}$ are $z / (a + k)$,
    so their polynomial normal form is $(z, a + k, 1)$. It follows that the
    sequence of polynomials
    \begin{equation*}
        p_d(k) = k^d (a + k - 1) - z (k + 1)^d
    \end{equation*}
    for $d = 0, 1, 2, \dots$ form a basis for the set of polynomials $p(k)$
    such that $p(k) z^k / a^{\overline{k}}$ is Gosper summable. It suffices to
    transform this basis by iteratively eliminating all powers of $k$ from
    $p_d(k)$ except its highest power and its constant term, then to show that
    the constant terms have the quoted exponential generating function.

    Note that
    \begin{equation*}
        p_d(k) = k^{d + 1} - (z + 1 - a) k^d - z \sum_{j < d} {d \choose j} k^j.
    \end{equation*}
    Therefore,
    \begin{equation*}
        k^{d + 1} + c(d + 1) = p_d(k) + (z + 1 - a)(k^d + c(d)) + z \sum_{0 < j < d} {d \choose j} (k^j + c(j)).
    \end{equation*}
    Comparing constant terms, we see that
    \begin{equation*}
        c(d + 1) = -z + (z + 1 - a) c(d) + z \sum_{0 < j < d} {d \choose j} c(j)
    \end{equation*}
    If we let $c(0) = -1$, then this becomes
    \begin{equation*}
        c(d + 1) = (1 - a) c(d) + z \sum_j {d \choose j} c(j).
    \end{equation*}
    If $C(x) = \sum_{d \geq 0} \frac{c(d)}{d!} x^d$ is the exponential
    generating function of $c(d)$, then the previous equation implies
    \begin{equation*}
        C'(x) = (1 - a) C(x) + z e^x C(x).
    \end{equation*}
    Solving this linear differential equation yields
    \begin{equation*}
        C(x) = -e^{-z - (a - 1) x + ze^x}.
    \end{equation*}
    Therefore $(k^d - c(d)) z^k / a^{\overline{k}}$ is Gosper summable if and
    only if $c(d)$ is the coefficient on $x^d / d!$ in $\exp(-z - (a - 1) x +
    ze^x)$.
\end{proof}

\begin{proposition}
    If $z \neq 0$ and $a$ is not a nonpositive integer, then the hypergeometric
    term $(k^d - c(d)) z^k a^{\overline{k}}$ is Gosper summable with respect to
    $k$ if and only if
    \begin{equation*}
        c(d) = [x^d / d!] \exp(z^{-1} - ax - z^{-1} e^{-x}) = [x^d / d!] g_{a, z}(x).
    \end{equation*}
\end{proposition}

\begin{proof}
    The consecutive term ratio of $z^k a^{\overline{k}}$ has polynomial normal
    form $(z(a + k), 1, 1)$. Therefore, as in the proof of Proposition~2, the
    sequence of polynomials
    \begin{equation*}
        p_d(k) = z (k + 1)^d (a + k) - k^d
    \end{equation*}
    form a basis for the set of all polynomials $p(k)$ such that $p(k) z^k
    a^{\overline{k}}$ is Gosper summable, and our job is to simplify it.

    Note that
    \begin{equation*}
        p_d(k) = zk^{d + 1} + (z(a + d) - 1)k^d + z \sum_{j < d} \left(a {d \choose j} + {d \choose j - 1}\right) k^j.
    \end{equation*}
    Therefore, having constructed basis elements of the form $k + c(1)$, $k^2 +
    c(2), \dots, k^d + c(d)$, we have
    \begin{equation*}
        z(k^{d + 1} + c(d + 1)) = p_d(k) - (z(a + d) - 1)(k^d + c(d))
        - z \sum_{0 < j < d}\left(a{d \choose j} + {d \choose j - 1}\right)(k^j + c(j))
    \end{equation*}
    Comparing constant coefficients yields
    \begin{equation*}
        z c(d + 1) = az - (z(a + d) - 1) c(d) - z \sum_{0 < j < d}\left(a{d \choose j} + {d \choose j - 1}\right)c(j).
    \end{equation*}
    If we let $c(0) = -1$, then this becomes
    \begin{equation*}
        z c(d + 1) = c(d) - z \sum_{0 \leq j \leq d}\left(a {d \choose j} + {d \choose j - 1}\right)c(j).
    \end{equation*}
    If we move the sum to the left-hand side, the equation reads
    \begin{equation*}
        c(d) = z\sum_j \left(a{d \choose j} + {d \choose j - 1}\right) c(j).
    \end{equation*}
    If $C(x)$ is the exponential generating function of $c(d)$, then standard
    techniques give us
    \begin{equation*}
        C(x) = z (e^x C(x) + e^x C'(x)),
    \end{equation*}
    whose unique solution with $C(0) = -1$ is $C(x) = -\exp(z^{-1} - ax -
    z^{-1}e^{-x})$.
\end{proof}

\section{Explicit Formulas and Gould Numbers}
\label{sec:explicit}

In the previous section we proved that
\begin{equation*}
    \sum_k \frac{k^d - b(d)}{k!}
\end{equation*}
is Gosper summable when $b(d)$ is the $d$th Bell number. In this section we
will explicitly evaluate this sum in terms of well-known integer sequences.

Equation~\eqref{basis1} is essentially a change of basis equation. It tells us
how to express the polynomials $p_d(k)$ in terms of the polynomials $k^d -
b(d)$. The first basis, $p_d(k)$, has the benefit that
\begin{equation*}
    \sum_k \frac{p_d(k)}{k!} = -\frac{k^{d + 1}}{k!}.
\end{equation*}
So, if we could invert \eqref{basis1} and express $k^d - b(d)$ in terms of
$p_d(k)$, $p_{d - 1}(k)$, and so on, we could apply linearity to evaluate
$\sum_k \frac{k^d - b(d)}{k!}$.

Equation~\eqref{basis1} amounts to the following matrix identity:
\begin{equation}
    \label{matrix}
    \begin{bmatrix}
        p_0(k) \\
        p_1(k) \\
        p_2(k) \\
        \vdots
    \end{bmatrix}
    =
    \begin{bmatrix}
        1 & 0 & 0 & 0 & \cdots \\
        -1 & 1 & 0 & 0 & \cdots \\
        -2 & -1 & 1 & 0 & \cdots \\
        -3 & -3 & -1 & 1 & \cdots \\
           \vdots
    \end{bmatrix}
    \begin{bmatrix}
        k - b(1) \\
        k^2 - b(2) \\
        k^3 - b(3) \\
        \vdots
    \end{bmatrix}.
\end{equation}
The coefficient matrix is invertible. The first few rows of $A^{-1}$ are as
follows:
\begin{equation*}
    A^{-1} =
\left[\begin{array}{cccccc}
        1 & 0 & 0 & 0 & 0 & \cdots \\
        1 & 1 & 0 & 0 & 0 & \cdots \\
        3 & 1 & 1 & 0 & 0 & \cdots \\
        9 & 4 & 1 & 1 & 0 & \cdots \\
        31 & 14 & 5 & 1 & 1 & \cdots \\
        \vdots
\end{array}\right].
\end{equation*}
The OEIS \cite{oeis} suggests that the columns are the diagonals of A121207,
which is a table of values $T_{dj}$ defined by
\begin{equation*}
    T_{(d + 1)j} = \sum_{i = 0}^{d - j - 1} {r \choose i} T_{(d - i)j}.
\end{equation*}
This table is a special case of a family of tables studied by Gould and
Quaintance \cite{gould}. The numbers $T_{r1}$ are called the \emph{Gould
numbers} (see \href{https://oeis.org/A040027}{A040027}). This suggestion turns
out to be correct.

\begin{proposition}
    For any positive integer $d$,
    \begin{equation*}
        \sum_k \frac{k^d - b(d)}{k!} = -\frac{\sum_{j \geq 1} B_{dj} k^j}{k!},
    \end{equation*}
    where the matrix $B_{dj}$ is defined by $B_{dd} = 1$ and
    \begin{equation*}
        B_{(d + 1)j} = \sum_{k = 0}^{d - j} {d \choose k} B_{(d - k)j} \quad (d \geq j).
    \end{equation*}
    In particular, $B_{dj}$ is the $d$th element of the $(j - 1)$th diagonal of
    \href{https://oeis.org/A121207}{A121207}.
\end{proposition}

\begin{proof}
    The matrix in \eqref{matrix} is defined by
    \begin{equation*}
        A_{dj} = [d = j] - {d - 1 \choose j} [d \neq j].
    \end{equation*}
    For $B$ to be its matrix inverse, we must have
    \begin{equation}
        \label{inv}
        \sum_{k \geq 1} A_{(d + 1)k} B_{kj} = [d + 1 = j]
    \end{equation}
    for all integers $d \geq 0$ and $j \geq 1$. If we expand $A_{(d + 1)k}$,
    then this reads
    \begin{equation}
        \label{brec}
        B_{(d + 1)j} = \sum_{k \geq 1} {d \choose k} B_{kj} + [d + 1 = j].
    \end{equation}
    Note that this implies $B_{dj} = 0$ if $d < j$. Indeed, $B_{1j} = {0
    \choose 1} B_{1k} = 0$, and if $d < j - 1$, then we can apply induction to
    every term of the right-hand side of \eqref{brec} to conclude that $B_{(d +
    1)j} = 0$. Hence we can define $B_{dj}$ as follows:
    \begin{align*}
        B_{jj} &= 1 \\
        B_{(d + 1)j} &= \sum_{j \leq k \leq d} {d \choose k} B_{kj} \quad (d \geq j).
    \end{align*}
    This is \href{https://oeis.org/A121207}{A121207} shifted so that $j$ begins
    at $1$ rather than $0$.
\end{proof}

Written more concretely, this identity reads
\begin{equation*}
    \sum_{k = 0}^{n - 1} \frac{k^d - b(d)}{k!} = -\frac{\sum_{j \geq 1} B_{dj} n^j}{n!}.
\end{equation*}
If we multiply by $n!$ and rearrange things, we obtain the following equality
for the bell numbers, valid for $n \geq 1$ and $d \geq 0$:
\begin{equation}
    b(d)
    =
    \frac{\sum_{k = 0}^{n - 1} k^d n^{\underline{n - k}} + \sum_{j \geq 1} B_{dj} n^j}{\sum_{k = 0}^{n - 1} n^{\underline{n - k}}}.
\end{equation}
It seems plausible that this has a combinatorial proof, but the author does not
know one.

\section{Connections with Bell numbers}
\label{sec:bell}

The exponential generating functions from the previous section are
\begin{align*}
    f_{a, z}(x) &= \exp(-z - (a - 1) x + ze^x) \\
    g_{a, z}(x) &= \exp(z^{-1} - a x - z^{-1}e^{-x}).
\end{align*}
These functions, and therefore the underlying sequences, are connected with the
Bell numbers. In particular, if we let
\begin{equation*}
    B(x) = e^{e^x - 1} = \sum_{j \geq 0} \frac{b(d)}{d!} x^d
\end{equation*}
be the exponential generating function for the Bell numbers, then for integral
$z$ we have the following identities:
\begin{align}
    \label{feqn} f_{a, z}(x) &= e^{(1 - a) x} B(x)^z \\
    \label{geqn} g_{a, 1 / z}(x) &= e^{-ax} B(-x)^{-z}.
\end{align}
If $z$ is positive, the first equation says that the coefficients of $f_{1,
z}(x)$ are the binomial convolution of $(1 - a)^k$ with the convolution of the
Bell numbers with themselves $z$ times. If $z$ is negative, the second equation
says that the coefficients of $g_{1, 1/z}(x)$ are the binomial convolution of
$(-a)^k$ with the convolution of the alternating Bell numbers $(-1)^d b(d)$
with themselves $z$ times.

\paragraph{Examples for $z^k / a^{\overline{k}}$} Setting $a = z = 1$ in
\eqref{feqn}, we get $f_{1, 1}(x) = B(x)$. If we translate this into the
vocabulary of the previous section, this says that
\begin{equation*}
    \frac{k^d - b(d)}{k!}
\end{equation*}
is Gosper summable, and no other constants will work. Similarly, $f_{1, 2}(x) =
B(x)^2$, so
\begin{equation*}
    \frac{(k^d - c(d)) 2^k}{k!}
\end{equation*}
is Gosper summable only if $c(d) = \sum_j {d \choose j} b(d) b(d - j)$. This
sequence begins $2, 6, 22, 94$, corresponding to the following identities:
\begin{align*}
    \sum_{k = 0}^n \frac{(k - 2)2^k}{k!} &= -\frac{2^{n + 1}}{n!} \\
    \sum_{k = 0}^n \frac{(k^2 - 6)2^k}{k!} &= -\frac{(n + 3) 2^{n + 1}}{n!} \\
    \sum_{k = 0}^n \frac{(k^3 - 22)2^k}{k!} &= -\frac{(n^2 + 4n + 11) 2^{n + 1}}{n!} \\
    \sum_{k = 0}^n \frac{(k^4 - 94)2^k}{k!} &= -\frac{(n^3 + 5n^2 + 17n + 47)2^{n + 1}}{n!}.
\end{align*}
Setting $a = 1/2$ and $z = 1$ gives $g_{1/2, 1}(x) = e^{x/2} B(x)$, which says
that
\begin{equation*}
    \frac{k^d - c(d)}{(1/2)^{\overline{k}}} = (k^d - c(d)) \frac{4^k k!}{(2k)!}
\end{equation*}
is Gosper summable only if $c(d) = \sum_j {d \choose j} b(d) / 2^{d - j}$.

\paragraph{Examples for $z^k a^{\overline{k}}$}
The connection for $g_{1, 1/z}(x)$ is most convenient when $z$ is a negative
integer. Setting $z = -1$ in \eqref{geqn} gives $g_{1, -1}(x) = e^{-x} B(-x) =
B'(-x)$, which says that
\begin{equation*}
    (k^d - (-1)^d b(d + 1)) (-1)^k k!
\end{equation*}
is Gosper summable, and no constant except $(-1)^d b(d + 1)$ will work.
Similarly, $g_{1, -1/2}(x) = e^{-x} B(-x)^2 = B'(-x) B(-x)$. Therefore,
\begin{equation*}
    (k^d - c(d)) \frac{k!}{(-2)^k}
\end{equation*}
is Gosper summable only if $c(d) = (-1)^d \sum_j {d \choose j} b(j + 1) b(d -
j)$. For example,
\begin{equation*}
    \sum_{k = 0}^n \frac{(k^2 - 11) k!}{(-2)^k} = \frac{(n - 3)(n + 1)!}{(-2)^n} - 8.
\end{equation*}
Setting $a = 1/2$ and $z = -1$ gives $g_{1/2, -1}(x) = e^{-x/2} B(-x)$, so
\begin{equation*}
    (k^d - c(d)) (-1)^k (1/2)^{\overline{k}} = (k^d - c(d)) (-1)^k \frac{(2k)!}{4^k k!}
\end{equation*}
is Gosper summable only if $c(d) = (-1)^d \sum_j {d \choose j} b(d) / 2^{d -
j}$.

\section{Conclusion}

We have given some explicit conditions for the Gosper summability of
hypergeometric terms of the form
\begin{equation*}
    (k^d - c(d)) z^k a^{\overline{k}} \quad \text{and} \quad (k^d - c'(d)) \frac{z^k}{a^{\overline{k}}}.
\end{equation*}
Namely, $c(d)$ and $c'(d)$ must be the coefficients of explicit exponential
generating functions which are related to the Bell numbers. In the special case
of $1 / k!$, we gave an explicit evaluation of these sums in terms of the
inverse of a matrix involving binomial coefficients. The Bell numbers
\emph{probably} appear by accident. However, should some combinatorial
connection be made, the author would like to hear about it.

We have made use of Gosper's algorithm for hypergeometric summation, but there
is a continuous variant of Gosper's algorithm for hyperexponential integration
\cite{az}. We may be able to make statements about when integrals of the form
\begin{equation*}
    \int (x^d - b(d)) e^{-x^2}\ dx
\end{equation*}
are themselves hyperexponential. However, in contrast to the summation problem,
we have a solid understanding of \emph{all} elementary antiderivatives thanks
to Liouville's theorem \cite[ch.~12]{alg}, not just hyperexponential ones. Thus
this could be a less satisfying problem.

Finally, we note that the results here work essentially because the space of
polynomials $p(k)$ such that $p(k) z^k a^{\overline{k}}$ are Gosper summable
contains polynomials of every degree greater than or equal to $1$. More
complicated hypergeometric terms will produce spaces with degrees only $2$ or
greater, or $3$ or greater, and so on. In these cases, the basis could not be
simplified down to leading powers and constants, so the results would be about
terms of the form
\begin{equation*}
    (k^d - k c_1(d) - c_0(d)) f(k),
\end{equation*}
or
\begin{equation*}
    (k^d - k^2 c_2(d) - k c_1(d) - c_0(d)) f(k),
\end{equation*}
and so on. The techniques here would certainly apply to such terms, though the
results would be more difficult to state.

\end{document}